\documentclass[12pt]{article}

\usepackage{amsmath,amsfonts}
\usepackage{amssymb}
\usepackage{amsthm}

\usepackage{times,txfonts}
\usepackage{mathrsfs}
\usepackage{bbm}
\usepackage{amsmath,amssymb}
\usepackage{color}

\usepackage[titletoc]{appendix}
\usepackage[english]{babel}
\makeatletter
\@addtoreset{equation}{section}
\makeatother

\def\@makecaption#1#2{\vskip\abovecaptionskip
  \sbox\@tempboxa{\small #1: #2}%
  \ifdim \wd\@tempboxa >\hsize \small #1: #2\par
  \else \global \@minipagefalse \hb@xt@\hsize{\hfil\box\@tempboxa\hfil}\fi
  \vskip\belowcaptionskip}
\newtheorem{theorem}{Theorem}[section]

\newtheorem{proposition}[theorem]{Proposition}
\newtheorem{corollary}[theorem]{Corollary}
\newtheorem{definition}[theorem]{Definition}

\theoremstyle{remark}

\def\bN{{\mathbb N}}

\newcommand\Z{\mathbb{Z}}
\renewcommand{\L}{\mathcal{L}}
\newcommand{\bbaa}{\begin{eqnarray}}
\newcommand{\eeaa}{\end{eqnarray}}
\newcommand{\mbb }[1]{\mathbb #1}
\newcommand{\nn}{\nonumber}

\def\[{\begin{eqnarray}}
\def\]{\end{eqnarray}}
\def\d{\partial}

\def\La{\Lambda}

\begin{document}

\begin{center}
{\large \bf Bosonic symmetries of the extended fermionic $(2N,2M)$-Toda hierarchy}\\[0.5cm]

{\large Chuanzhong Li }

{}~\\
\quad \\

{{Department of Mathematics, Ningbo University, Ningbo, 315211 Zhejiang, P.\ R.\ China}}\\
\em { Email: lichuanzhong@nbu.edu.cn}\\

\end{center}

\centerline{{\bf Abstract}} \noindent
In this paper,  we construct the additional symmetries of the  fermionic $(2N,2M)$-Toda hierarchy basing on the generalization of the $N{=}(1|1)$ supersymmetric two dimensional Toda lattice hierarchy.
These additional  flows constitute a $w_{\infty}\times w_{\infty}$ Lie algebra. As a Bosonic reduction of the   $N{=}(1|1)$ supersymmetric two dimensional Toda lattice hierarchy and the  fermionic $(2N,2M)$-Toda hierarchy, we define a new extended fermionic $(2N,2M)$-Toda hierarchy which admits a Bosonic Block type superconformal structure.



\tableofcontents

\pagenumbering{arabic}

\section{Introduction}

In the study of  integrable hierarchies, it  is interesting
to identify their symmetries and
the algebraic structure of the symmetries, particularly to find the additional symmetry.
Additional symmetries of the Kadomtsev-Petviashvili(KP) hierarchy
were introduced  by Orlov and Shulman \cite{os1} which contain one important symmetry called Virasoro symmetry.
 As a nonlinear evolutionary
differential-difference equation
describing an infinite system of masses on a line that interact
through an exponential force, the Toda equation\cite{Toda} was  generalized to Toda lattice hierarchy \cite{UT} which is completely integrable and has important applications in many fields such as classical and quantum field theory, in
particular in the theory of Gromov-Witten invariants (\cite{Z}).
Adding extended logarithmic flows, the Toda lattice
hierarchy was extended into the so-called extended Toda hierarchy
(\cite{CDZ}) and was conjectured and  shown (\cite{DZ},
\cite{Ge}, \cite{Z}) that the extended Toda hierarchy is a
hierarchy describing the Gromov-Witten invariants of $CP^1$ as the large $N$
limit of the $CP^1$ topological sigma model.   The extended bigraded
Toda hierarchy (EBTH) was introduced by Guido Carlet (\cite{C}) who
hoped that EBTH might also be relevant for some applications in two dimensional(2D)
topological field theory and in the theory of Gromov-Witten
invariants.
The Hirota bilinear equation of EBTH were equivalently constructed in our early paper \cite{ourJMP} and a very recent paper \cite{leurhirota}, because of the equivalence of $t_{1,N}$ flow and $t_{0,N}$ flow of the EBTH in \cite{ourJMP}. Meanwhile it was proved to govern Gromov-Witten invariant of the total
  descendent potential of $\mathbb{P}^1$ orbifolds \cite{leurhirota}. This
hierarchy also attracted a series of results from analytical and algebraic considerations
\cite{ourBlock,solutionBTH,RMP,BKP-DS,torus}.

Various generalizations and supersymmetric extensions \cite{yamada} of the KP hierarchy have
deep implications in mathematical physics, particularly in the theory of Lie algebras.  One important supersymmetric extension is the supersymmetric Manin-Radul Kadomtsev-Petviashvili
(MR-SKP) hierarchy\cite{maninsuperKP}  which contains a lot of reduced integrable super solitary equations including the Kupershmidt's super-KdV equation.
The additional symmetry of the MR-SKP hierarchy was studied by Stanciu \cite{Stanciu}.   Later the supersymmetric BKP (SBKP) hierarchy  was constructed in \cite{StanciuBKP}. In our quite recent paper\cite{NPB}, we  constructed the additional symmetries of the supersymmetric BKP hierarchy.
These additional flows constitute a B type $SW_{1+\infty}$ Lie algebra. Further we generalize the SBKP hierarchy to a supersymmetric two-component BKP hierarchy (S2BKP) hierarchy
and derive its algebraic structure. As a reduction of the S2BKP hierarchy,  a new supersymmetric Drinfeld-Sokolov hierarchy of type D was constructed and proved to have a super Block type additional symmetry. Later the Darboux transformation of the supersymmetric BKP (SBKP) hierarchy was constructed in \cite{JNMP}.

The supersymmetric Toda hierarchy was introduced in \cite{Ik}.  By introducing the Lie superalgebra
osp$(\infty|\infty)$, the ortho-symplectic supersymmetric Toda hierarchy was defined as well in \cite{Ik2}. These equations in the hierarchy were solved through
the Riemann-Hilbert decomposition of corresponding infinite dimensional Lie supergroups.
Recently, the supersymmetric Toda hierarchy was studied a lot including their hamiltonian structure, dispersionless limit and so on \cite{LSor1}-\cite{KS2}. Recently, about the application of Block algebra in KP and Toda type systems, we also prove that the constrained BKP system and Toda system all have a Block symmetry in \cite{blockds,IJGMMP}.

This paper is arranged as follows. In the next section, some fundamental notations on  the space of supersymmetric shift operators will be prepared. After this, we introduce some
necessary facts of the  $N{=}(1|1)$ supersymmetric Toda hierarchy in Section 3. In Sections 4, we will give an introduction of the fermionic $(K,S)$-Toda hierarchy and
we will give the additional symmetries for the
  fermionic $(2N,2M)$-Toda hierarchy
and derive its  $w_{\infty}\times w_{\infty}$  symmetry in Section 5.
 As a Bosonic reduction of the  fermionic $(2N,2M)$-Toda hierarchy, we define a new constrained system called the extended  fermionic $(2N,2M)$-Toda hierarchy in Section 6 and prove that it possesses a Bosonic Block type Lie symmetry in Section 7.

\section{ Supersymmetric shift operators}
 In this section, we will recall an introduction about the space of supersymmetric shift operators following a series of papers of Sorin's group, for example the reference \cite{KS2}. Firstly  the space of difference operators
 can be represented in the following general form
 \bbaa  \label{O}
\mbb{F}_m=\sum\limits_{k=-\infty}^{\infty}f^{(m)}_{k}(x)\ \Lambda^{k-m}, \ \ \  \ \Lambda=e^{\epsilon\partial},\ \ \ m, k \in \mathbb{Z},
 \eeaa
where  the parameter $\epsilon$ in the class of difference operators is the string coupling constant and
 $\Lambda$ is the shift operator whose action on the continuous spatial  fields
as a shift of a lattice index $\epsilon$
 \bbaa \label{shift}
 \Lambda^{l}f^{(m)}_{k}(x)=f^{(m)}_{k}(x+l\epsilon)\Lambda^{l}.
 \eeaa
 The functions $f_{2k}(x)^{(m)}$  ($f_{2k+1}(x)^{(m)}$) are
the  bosonic (fermionic) lattice fields and the Grassmann $Z_2$-parity can be defined
\bbaa \label{f-grad} d_{f^{(m)}_{k}(x)}=|k|\ \mbox{mod}\ 2,\ \  d_{\Lambda^l}=|l|\ \mbox{mod} \ 2.\nn \eeaa
The operators  $ \mbb{F}_m$ in \eqref{O} admit the diagonal $Z_2$-parity
\bbaa\label{Z2-par}
d_{\mbb{F}_m}=d_{f_{k}(x)^{(m)}}+d_{\Lambda^{k-m}}=|m|\
\mbox{mod} \ 2 .\eeaa

The involution $*$ can be defined as \cite{Ik}
\begin{eqnarray}
{\mbb{F}}^{*}_m=\sum^{\infty}_{k=-\infty} (-1)^k
f^{(m)}_{k}(x)\Lambda^{k-m}. \label{invtoda}
\end{eqnarray}

 In the follows we need the projections
 of the operators $\mbb{F}_m$ defined as

\bbaa  ( \mbb{F}_m)_{\leqslant p}=\sum\limits_{k\leqslant
p+m}f_{k}(x)\ \Lambda^{k-m},\ \ \ \
 (\mbb{F}_m)_{\geqslant p}=\sum\limits_{k\geqslant p+m}f_{k}(x)\
 \Lambda^{k-m},\nn
 \eeaa

and we will use the usual notation for the projections
$(\mbb{F}_m)_+:=(\mbb{F}_m)_{\geqslant 0}$ and
$(\mbb{F}_m)_-:=(\mbb{F}_m)_{< 0}$.

 The generalized graded  algebra on these
subspaces will be defined with the following bracket
\cite{KS2}
\bbaa\label{SK-bracket}  [{\mathbb F}, \widetilde{ \mathbb F}
\}:={\mathbb F}\ \widetilde{\mathbb F}
 - (-1)^
 {d_{{\mathbb F}
 \vphantom{\widetilde{\mathbb F}} }
 d_{\widetilde{\mathbb F}}}~{\widetilde{\mathbb F}}^{*(d_{{\mathbb F}\vphantom{\widetilde{\mathbb F}}})}~ {{\mathbb F}}^{*(d_{\widetilde{\mathbb F}})}, \eeaa
where $\mathbb F^{*(j)}$ denotes the j-fold action of the involution $*$ on the operator $\mathbb F$ with $\mathbb F^{*(2)}=\mathbb F.$ The bracket \eqref{SK-bracket} generalizes  the
(anti)commutator in superalgebras and satisfies the following
properties \cite{KS2} including
symmetry, derivation  and  Jacobi identity as
 \bbaa \label{symSK} [ {\mathbb F}, \widetilde{ \mathbb F}
\}=
 - (-1)^
 {d_{{\mathbb F}
 \vphantom{\widetilde{\mathbb F}} }
 d_{\widetilde{\mathbb F}}}~[{\widetilde{\mathbb F}}^{*(d_{{\mathbb F}\vphantom{\widetilde{\mathbb F}}})},  {{\mathbb F}}^{*(d_{\widetilde{\mathbb F}})}\}, \eeaa

 \bbaa \label{derSK}
 [{\mathbb F}, \widetilde{ \mathbb F}\, \widehat{ \mathbb F}
\}=[ {\mathbb F}, \widetilde{ \mathbb F}\}\, \widehat{ \mathbb F}
 + (-1)^{d_{{\mathbb F}
 \vphantom{\widetilde{\mathbb F}} } d_{\widetilde{\mathbb F}}}\,
 {\widetilde{\mathbb F}}^{*(d_{{\mathbb F}\vphantom{\widetilde{\mathbb F}}})} [{\mathbb F}^{*(d_{\widetilde{\mathbb F}})}, {\widehat{\mathbb F}}\},\eeaa

\bbaa\label{SK-Jacobi} (-1)^{d_{{\mathbb F}
 \vphantom{\widetilde{\mathbb F}} } d_{\widehat{\mathbb F}\vphantom{\widetilde{\mathbb
 O}}}}\,
 [[{{\mathbb F}},\, {\widetilde{\mathbb F}}^{*(d_{{\mathbb F}\vphantom{\widetilde{\mathbb F}}})}\},\,  {\widehat{\mathbb F}}^{*(d_{{\mathbb F}\vphantom{\widetilde{\mathbb F}}}+
d_{\widetilde{\mathbb F}}) \vphantom{\widetilde{\mathbb F}}}\}
 +(-1)^{d_{\widetilde{{\mathbb F}}
 \vphantom{\widetilde{\mathbb F}} } d_{{\mathbb F}\vphantom{\widetilde{\mathbb
 O}}}}\,
 [[{\widetilde{\mathbb F}},\, {\widehat{\mathbb F}}^{*(d_{\widetilde{{\mathbb F}}\vphantom{\widetilde{\mathbb F}}})}\},\,
 {{\mathbb F}}^{*(d_{\widetilde{{\mathbb F}}\vphantom{\widetilde{\mathbb F}}}+d_{\widehat{\mathbb F}}) \vphantom{\widetilde{\mathbb F}}}\}&&\nn\\
 +\ \  (-1)^{d_{\widehat{\mathbb F}
 \vphantom{\widetilde{\mathbb F}} } d_{\widetilde{\mathbb F}\vphantom{\widetilde{\mathbb
 O}}}}\,
[[{\widehat{\mathbb F}},\, {{\mathbb F}}^{*(d_{\widehat{\mathbb F}\vphantom{\widetilde{\mathbb F}}})}\},\,
 {\widetilde{\mathbb F}}^{*(d_{\widehat{{\mathbb F}}\vphantom{\widetilde{\mathbb F}}}+d_{{\mathbb F}}) \vphantom{\widetilde{\mathbb F}}}\}
 &=&0.
\eeaa

Same as \cite{Ik},   the operator $\mathbb F^k_*$ is defined as following

\bbaa
\mathbb F^{2k}_*=(\mathbb F^{*(d_\mathbb F)}\mathbb F)^k,\ \ \mathbb F^{2k+1}_*=\mathbb F\mathbb F^{2k}_*.\eeaa

With the above preparation, the $N{=}(1|1)$ supersymmetric 2D Toda lattice hierarchy will be  introduced as \cite{superbigraded} in the next section.

\section{$N{=}(1|1)$ supersymmetric 2D Toda lattice hierarchy}

The Lax operators $L,\bar L$ of the $N{=}(1|1)$ supersymmetric 2D Toda lattice hierarchy belong to
the space of operators
\begin{eqnarray}
L=\sum^{\infty}_{k=0} u_{k}(x)\Lambda^{1-k},\quad
u_{0}(x)= 1, \quad \bar L=\sum^{\infty}_{k=0}
v_{k}(x)\Lambda^{k-1}, \quad v_{0}(x)\neq 0 \label{laxs1}
\end{eqnarray}
and have the grading $d_{ L}=d_{\bar L}=1$. The coefficient functions in the lax operators depend on infinitely many
dependent variables.
Then the Lax operators $L,\bar L$ can have the following dressing structures
\[L=W^*\Lambda W^{-1},\ \ \bar L=\bar W^*\Lambda^{-1} \bar W^{-1},\]
where
\[W=\sum^{\infty}_{k=0} W_{k}(x)\Lambda^{-k},\quad
W_{0}(x)= 1,\ \  \ \bar W=\sum^{\infty}_{k=0}\bar W_{k}(x)\Lambda^{k}. \]
Here $d_{W}=d_{\bar W}=0\ \mbox{mod} \ 2 .$
Then the Lax operators $L_*^{2n},\bar L_*^{2m}$ can have the following dressing structures
\[L_*^{2n}&:=&(L^*L)^{n}=(W\Lambda W^{*-1}W^*\Lambda W^{-1})^n=W\Lambda^{2n} W^{-1},\\
 \ \bar L_*^{2m}&:=&(\bar L^*\bar L)^m=(\bar W\Lambda^{-1} \bar W^{*-1}\bar W^*\Lambda^{-1} \bar W^{-1})^m=\bar W\Lambda^{-2m} \bar W^{-1},\]
and the Lax operators $L_*^{2n+1},\bar L_*^{2m+1}$ can have the following dressing structures
\[L_*^{2n+1}&:=&L(L^*L)^{n}=W^*\Lambda^{2n+1} W^{-1},\\
 \ \bar L_*^{2m+1}&:=&\bar L(\bar L^*\bar L)^m=\bar W^*\Lambda^{-2m-1} \bar W^{-1}.\]

The Lax
representation of the $N{=}(1|1)$ supersymmetric 2D Toda lattice hierarchy in terms of bracket
operations is as following \cite{superbigraded}
\begin{eqnarray}
D_n L&=& - (-1)^n
\Bigr[{((L^{n}_{*})}_{- })^{*},
L\Bigl\}, \quad n \in {\bN},
\label{laxreprs1}
\end{eqnarray}
\begin{eqnarray}
D_n \bar L&=&  (-1)^n
\Bigr[{((L^{n}_{*})}_{+ })^{*},
\bar L\Bigl\},  \quad n \in {\bN},
\label{laxreprs12}
\end{eqnarray}
\begin{eqnarray}
\bar D_n L&=&  -  (-1)^n
\Bigr[{((\bar L^{n}_{*})}_{- })^{*},
L\Bigl\}, \quad n \in {\bN},
\label{laxreprs21}
\end{eqnarray}
\begin{eqnarray}
\bar D_n \bar L&=& (-1)^n
\Bigr[{((\bar L^{n}_{*})}_{+ })^{*},
\bar L\Bigl\},  \quad n \in {\bN},
\label{laxreprs2}
\end{eqnarray}
where $ D_{2n}$ ($ D_{2n+1}$) and $\bar D_{2n}$ ($\bar D_{2n+1}$) are bosonic (fermionic) evolution
derivatives.

The Lax equations  generate a non-Abelian algebra of flows
of the $N{=}(1|1)$ supersymmetric 2D Toda lattice hierarchy,

\begin{eqnarray}
[D_{n}~,~\bar D_{l}\}&=&[ D_{n}~,~ D_{2l}\}=[\bar D_{n}~,~\bar D_{2l}\}=0,\\
 \{ D_{2n+1}~,~ D_{2l+1}\}&=&2 D_{2(n+l+1)},\ \
\{\bar D_{2n+1}~,~\bar D_{2l+1}\}=2\bar D_{2(n+l+1)},
\label{algebras1}
\end{eqnarray}
which can be realized as
\begin{eqnarray}
 D_{2n} ={ \partial}_{2n}, \quad  D_{2n+1}
={ \partial}_{2n+1}+
\sum^{\infty}_{l=1} t_{2l-1}{ \partial}_{2(n+l)},\   \d_{m}=\d_{t_m},
\end{eqnarray}
\begin{eqnarray}
\bar D_{2n} ={\bar \partial}_{2n}, \quad \bar D_{2n+1}
={\bar \partial}_{2n+1}+
\sum^{\infty}_{l=1}\bar t_{2l-1}{\bar \partial}_{2(n+l)},\  \bar \d_{m}=\d_{\bar t_m},
\label{covder}
\end{eqnarray}
where $ t_{2n}$ ($ t_{2n+1}$), $\bar t_{2n}$ ($\bar t_{2n+1}$)
are bosonic (fermionic) evolution times.

Similarly to the Lax operators $L,\bar L$, the operators
$L^{n}_{*},\bar L^{n}_{*}$  can be represented as
\begin{eqnarray}
L^{m}_{*}:=\sum^{\infty}_{k=0}
u^{(m)}_{k}(x)\Lambda^{m-k}, \quad u^{(m)}_{0}(x)=1, \quad
\bar L^{m}_{*}:=\sum^{\infty}_{k=0}
v^{(m)}_{k}(x)\Lambda^{k-m}. \label{laxs2}
\end{eqnarray}
Here $u^{(m)}_{k}(x)$ and $v^{(m)}_{k}(x)$  are functionals of the original
fields $\{u_{k}(x),~v_{k}(x)\}$. The  $Z_2$-grading of the operator
$L^{n}_{*}, \bar L^{n}_{*}$ has the form $d_{ L^{2n}_{*}}=d_{\bar L^{2n}_{*}}=0$ and
$d_{ L^{2n+1}_{*}}=d_{\bar L^{2n+1}_{*}}=1$. This agrees with another $Z_2$-grading
$d_{ D_{2n}}=d_{\bar D_{2n}}=0$ and $d_{ D_{2n+1}}=d_{\bar D_{2n+1}}=1$
that corresponds to the statistics of the evolution derivatives
$D_{n},\bar D_{n}$.

Using  the definitions of $L^{n}_{*},\bar L^{n}_{*}$, we can easily obtain
the following useful identities \cite{superbigraded}
\begin{eqnarray}
&&\Bigr[ L^{2n}_{*},
L^{2m}_{*}\Bigl\}=0, \nonumber\\
&&\Bigr[ (L^{2n}_{*})^{*},
L^{2m+1}_{*}\Bigl\}=0, \quad \Bigr[
L^{2n+1}_{*},
L^{2m}_{*}\Bigl\}=0, \nonumber\\
&&\Bigr[ (L^{2n+1}_{*})^{*},
L^{2m+1}_{*}\Bigl\}=2L^{2(n+m+1)}_{*},
\end{eqnarray}
\begin{eqnarray}
&&\Bigr[ \bar L^{2n}_{*},
\bar L^{2m}_{*}\Bigl\}=0, \nonumber\\
&&\Bigr[ (\bar L^{2n}_{*})^{*},
\bar L^{2m+1}_{*}\Bigl\}=0, \quad \Bigr[
\bar L^{2n+1}_{*},
\bar L^{2m}_{*}\Bigl\}=0, \nonumber\\
&&\Bigr[ (\bar L^{2n+1}_{*})^{*},
\bar L^{2m+1}_{*}\Bigl\}=2\bar L^{2(n+m+1)}_{*}.
\label{ident}
\end{eqnarray}
Next, using (\ref{laxreprs1}--\ref{laxreprs2}), one can derive dynamical equations  for
the operators $ L^{n}_{*},\bar L^{n}_{*}$ \cite{superbigraded},
\begin{eqnarray}
 D_n L^{m}_{*}&=& -  (-1)^{nm} \Bigr[
{(( L^{n}_{*})}_{-})^{*(m)},
 L^{m}_{*}\Bigl\}, \label{laxreprs2m}
\end{eqnarray}
\begin{eqnarray}
 D_n \bar L^{m}_{*}&=&  (-1)^{nm} \Bigr[
{(( L^{n}_{*})}_{+})^{*(m)},
\bar L^{m}_{*}\Bigl\}, \label{laxreprs2m}
\end{eqnarray}
\begin{eqnarray}
\bar D_n L^{m}_{*}&=&   -(-1)^{nm} \Bigr[
{((\bar L^{n}_{*})}_{- })^{*(m)},
L^{m}_{*}\Bigl\}, \label{laxreprs22m}
\end{eqnarray}
\begin{eqnarray}
\bar D_n (\bar  L)^{m}_{*}&=&   (-1)^{nm} \Bigr[
{((\bar L^{n}_{*})}_{+ })^{*(m)},
\bar L^{m}_{*}\Bigl\}. \label{laxreprs22m}
\end{eqnarray}

The supersymmetric $N{=}(1|1)$ supersymmetric 2D Toda lattice equation
\begin{eqnarray}
D_1\bar D_1 \ln v_{0}(x)= v_{0}(x+\epsilon) - v_{0}(x-\epsilon) \label{todabs}
\end{eqnarray}
belongs to system of equations (\ref{laxreprs1}--\ref{laxreprs2}).

\section{ 2D   fermionic $(K,S)$-Toda lattice hierarchy}
 In this section, we construct the
two-dimensional  fermionic $(K,S)$-Toda lattice hierarchy in
terms of the Lax-pair representation similarly as the  fermionic $(K,S)$-Toda lattice hierarchy in \cite{superbigraded}.

 Let us consider two difference operators
 \bbaa\label{LaxOp} L_{K}=L^{K}=\sum\limits_{k=0}^\infty u_{Kk}(x)
\Lambda^{K-k}, \ \ \ \ \ \ \
\bar L_{S}=\bar L^{S}=\sum\limits_{k=0}^\infty v_{Sk}(x) \Lambda^{k-S},
 \eeaa
 which have the Grassmann parities as $d_{L_{K}}=K,\ \ d_{\bar L_{S}}=S.$
 If $K=2N, S=2M$, then the Lax operators $L^{2N},\bar L^{2M}$ can have the following dressing structures
\[L_{K}=L^{2N}&:=&(L^*L)^{N}=W\Lambda^{2N} W^{-1},\\
 \ L_{S}=\bar L^{2M}&:=&(\bar L^*\bar L)^M=\bar W\Lambda^{-2M} \bar W^{-1}.\]

The dynamics of the fields $u_{Kk}(x),v_{Sk}(x)$ are governed by the
Lax equations expressed in terms of the generalized graded bracket
\eqref{SK-bracket} \cite{KS2}
 \bbaa\label{LaxEqCom}
D_s(L_K)^r&=&-(-1)^{sr K }[(((L_K)^s)_-)^{*(rK)},(L_K)^r\},\nn\\
D_s(\bar L_S)^r&=&(-1)^{s r K S }[(((L_K)^s)_+)^{*(rS)},(\bar L_S)^r\},\nn\\
\bar D_s(L_K)^r&=&-(-1)^{sr K S}[(((\bar L_S)^s)_-)^{*(rK)},(L_K)^r\},\nn\\
\bar D_s(\bar L_S)^r&=&(-1)^{srS}[(((\bar L_S)^s)_+)^{*(rS)},(\bar L_S)^r\},
   \ \ \ \ \ s \in \mathbb N.
 \eeaa
The $Z_2$-parity of two kinds of derivatives is
defined as
$$ d_{D_s}=sK\ \mbox{ mod } \ 2,\ \
\ \ \ \ \ d_{\bar D_s}=sS\ \mbox{mod } \ 2.
 $$
The Lax equations  generate a non-Abelian  super algebra
of flows of the 2D  fermionic $(K,S)$- Toda lattice hierarchy
\bbaa [D_s,D_p\}=(1-(-1)^{spK})D_{s+p},\ \ \ \
\bigl[\bar D_s,\bar D_p\}   =  (1-(-1)^{spS})\bar D_{s+p},\ \ \ \
 \bigl[ D_s,\bar D_p\} =   0.\nn
 \eeaa
If we
introduce the notation
$v_{S0}(x)=\alpha(x),v_{S1}(x)=\rho(x),u_{K1}(x)=\gamma(x), u_{K2}(x)=\beta(x)$ and
 consider eqs. (\ref{LaxEqCom}) at  $K=S=2$, $r=s=1$.
One obtains
\bbaa\label{2dToda} &&D_1\alpha(x)=\alpha(x)(\beta(x)-\beta(x-2\epsilon)),\ \ \ \ \ \
\bar D_1\gamma(x)=\rho(x)u_0(x-\epsilon)-\rho(x+2\epsilon)u_{0}(x),
\nn \\
&& \bar D_1\beta(x)=\alpha(x)u_0(x-2\epsilon) -\alpha(x+2\epsilon)u_{0}(x)-\gamma(x)
\rho(x+\epsilon)-\gamma(x-\epsilon)\rho(x),\nn\\
&& D_1\rho(x)=\rho(x)(\beta(x)-\beta(x-\epsilon))+\alpha(x+\epsilon)
\gamma(x)-\alpha(x)\gamma(x-2\epsilon),\ \bar D_1u_{K0}(x)=0. \eeaa
It is easy to check that after a reduction $u_{K0}(x)=1$ eqs.(\ref{2dToda}) will become the 2D generalized  fermionic Toda
lattice equations as

\bbaa\label{2dToda2} &&D_1\alpha(x)=\alpha(x)(\beta(x)-\beta(x-2\epsilon)),\ \ \ \ \ \
\bar D_1\gamma(x)=\rho(x)-\rho(x+2\epsilon),
\nn \\
&& \bar D_1\beta(x)=\alpha(x) -\alpha(x+2\epsilon)-\gamma(x)
\rho(x+\epsilon)-\gamma(x-\epsilon)\rho(x),\nn\\
&& D_1\rho(x)=\rho(x)(\beta(x)-\beta(x-\epsilon))+\alpha(x+\epsilon)
\gamma(x)-\alpha(x)\gamma(x-2\epsilon). \eeaa

\subsection{ 2D   fermionic $(2N,2M)$-Toda hierarchy}
When $K=2N,S=2M$, then the Lax operators $L_{2N},\bar L_{2M}$ will be defined to satisfy the following specific flow equations of the 2D   fermionic $(2N,2M)$-Toda hierarchy
\begin{eqnarray}
 D_{2nN} L_{2N}&=& - \Bigr[
(L^{n}_{2N})_-,
 L_{2N}\Bigl\}, \ \
 D_{2nN} \bar L_{2M}=   \Bigr[
( L_{2N}^{n})_{+},
\bar L_{2M}\Bigl\}, \\
\bar D_{2nM} L_{2N}&=&   -\Bigr[
(\bar L_{2M}^n)_{- },
L_{2N}\Bigl\}, \ \
\bar D_{2nM} \bar  L_{2M}=    \Bigr[
(\bar L^{n}_{2M})_+,
\bar L_{2M}\Bigl\}. \label{laxreprs22}
\end{eqnarray}

In the next section, we will study on the  additional
symmetries of this  2D   fermionic $(2N,2M)$-Toda hierarchy.

\section{Additional symmetry of 2D   fermionic $(2N,2M)$-Toda hierarchy}

In this section, we will consider the construction of the flows of additional
symmetries of the  2D   fermionic $(2N,2M)$-Toda hierarchy.

Firstly, we introduce Orlov-Schulman operators as following
\begin{eqnarray}\label{Moperator}
&&M_N=W\Gamma_N W^{-1}, \ \ \bar M_M=\bar W\bar \Gamma_M \bar W^{-1},\ \\
 &&\Gamma_N=
\frac{x}{2N\epsilon}\Lambda^{-2nN}+\sum_{n\geq 0}
n\Lambda^{2N(n-1)}t_{2nN},\ \ \bar \Gamma_M=
\frac{-x}{2M\epsilon}\Lambda^{2nM}-\sum_{n\geq 0}
n\Lambda^{-2M(n-1)}\bar t_{2nM}.
\end{eqnarray}

Then one can prove the Lax operator $L_{2N},\bar L_{2M}$  and Orlov-Schulman operators $M_N,\bar M_M$ all have the even Grassmann parity and satisfy the following theorem.
\begin{proposition}\label{flowsofM}
The Lax operators $L_{2N},\bar L_{2M}$ and Orlov-Schulman operators $M_N,\bar M_M$ of the   fermionic $(2N,2M)$-Toda hierarchy
satisfy the following
\begin{eqnarray}
[L_{2N},M_N\}&=&1,[\bar L_{2M},\bar M_M\}=1,\\ \label{Mequation}
D_{2nN}(M_N)^r&=&[((L_{2N})^{n})_+,(M_N)^r\},\nn\\
D_{2nN}(\bar M_M)^r&=&[((L_{2N})^{n})_+,(\bar M_M)^r\},\nn\\
\bar D_{2nM}(M_N)^r&=&-[((\bar L_{2M})^{n})_-,(M_N)^r\},\nn\\
\bar D_{2nM}(\bar M_M)^r&=&-[((\bar L_{2M})^{n})_-,(\bar M_M)^r\},
   \ \ \ \ \ n \in \mathbb N.
 \end{eqnarray}
\end{proposition}

\begin{proof}
One can prove the proposition by dressing the following several commutative Lie brackets
\begin{eqnarray*}&&[ D_{ t_{2nN}}-\Lambda^{2nN},\Gamma_N\}\\
&=&[ D_{ t_{2nN}}-\Lambda^{2nN},\frac{x}{2N\epsilon}\Lambda^{-2N}+\sum_{n\geq 0}
n\Lambda^{2N(n-1)}t_{n}\}\\&=&0,
\end{eqnarray*}

\begin{eqnarray*}&&[ D_{ t_{2nN}},\bar \Gamma_M\}=[ D_{ t_{2nN}},\frac{-x}{2M\epsilon}\Lambda^{2nM}\}=0.
\end{eqnarray*}
The other identities can be proved in similar ways.
\end{proof}

For the
additional
symmetries of the   fermionic $(2N,2M)$-Toda hierarchy, we introduce additional
independent variables $t_{m,l},\bar t_{m,l}$ and define the actions of the
additional flows on the wave operators as
\begin{eqnarray}
\epsilon D_{t_{m,l}}W&=&-(M_N^mL_{2N}^l)_{-}W, \ \ \ \epsilon D_{\bar t_{m,l}}
\bar W=(\bar M_M^m\bar L_{2M}^l)_{+}\bar W,\\
\epsilon D_{t_{m,l}}\bar W&=&(M_N^mL_{2N}^l)_{+}\bar W, \ \ \ \epsilon D_{\bar t_{m,l}}
 W=-(\bar M_M^m\bar L_{2M}^l)_{-} W,
\end{eqnarray}
where $m\geq 0, l\geq 0$.

By considering the well-known $w_{\infty}$ algebra with algebraic coefficients $C_{\alpha\beta}^{(ps)(ab)}$ as
\begin{eqnarray*}
&&[z^s\partial^p,z^b\partial^a]=\sum_{\alpha\beta}C_{\alpha\beta}^{(ps)(ab)}z^{\beta}\partial_{\alpha},
\end{eqnarray*}
we can prove the above additional flows are in fact symmetries of the  fermionic $(2N,2M)$-Toda hierarchy whose structure is a $w_{\infty}\times w_{\infty}$ algebra as in the following theorem.
\begin{theorem}\label{WinfiniteCalgebra1}
The additional flows  $D_{t_{m,l}}, D_{\bar t_{m,l}}$ of the  fermionic $(2N,2M)$-Toda hierarchy form a  $w_{\infty}\times w_{\infty}$ Lie algebra with the
following relation
 \begin{eqnarray}\label{algebra relation}
[D_{t_{m,l}},D_{t_{n,k}}\}&=& \sum_{\alpha\beta}C_{\alpha\beta}^{(ml)(nk)}D_{ t_{\alpha,\beta}},\end{eqnarray}
\begin{eqnarray}[D_{\bar t_{m,l}},D_{\bar t_{n,k}}\}&=& \sum_{\alpha\beta}C_{\alpha\beta}^{(ml)(nk)}D_{\bar t_{\alpha,\beta}},
\end{eqnarray}
\begin{eqnarray}[D_{ t_{m,l}},D_{\bar t_{n,k}}\}&=& 0,
\end{eqnarray}
which holds in the sense of acting on  $W$, $\bar W$ or $L_{2N},\bar L_{2M}$ and  $m,n,l,k\geq 0.$
\end{theorem}

\begin{proof}
The proof can be proved similarly as the additional symmetry of two dimensional Toda hierarchy in \cite{adler}. Here we skip the proof.
\end{proof}

\section{ Extended  fermionic $(2N,2M)$-Toda hierarchy}
 In this section, we consider the reduction of the 2D
fermionic $(2N,2M)$-Toda lattice hierarchy for even values of
$(2N,2M)$ to the 1D space with additional logarithmic flows.

For even $(2N,2M)$, one can impose the
reduction constraint
 on the Lax operators \eqref{LaxOp} as follows:
\bbaa \label{redLax}
L^{2N}=\bar L^{2M}\equiv\L, \eeaa
which  leads to the following explicit form for the reduced Lax
operator
\bbaa \label{LaxKM} \L=\sum\limits_{k=0}^{2N+2M}
u_{k}(x) \Lambda^{2N-k}.
 \eeaa
One can find $\L^k_*=\L^k,\ \ k\in \Z_+$.
The Lax equation of the one dimensional fermionic $(2N,2M)$-Toda hierarchy on the reduced Lax
operator
\bbaa \label{1dLaxEq}\epsilon D_{t_s}\L&=&
[(\frac{\L^{s+1}}{(s+1)!})_+,\L\}, \eeaa
where the Lax operator $\L$ has the following dressing structure

\[\label{two dressing2}\L=W\Lambda^{2N} W^{-1}=\bar W\Lambda^{-2M} \bar W^{-1}.\]
Similarly as the \cite{C,ourJMP}, we define the following logarithmic operator
\begin{align}
\log_+\L&=2NW\circ\epsilon\partial\circ W^{-1},\\
\log_-\L&=-2M\bar W\circ\epsilon\partial\circ \bar W^{-1},
\end{align}
where $\d$ is the derivative about spatial variable $x$.
Combining these above logarithm operators together can derive the following important logarithmic operator
\begin{align}
\label{Log} \log \L:&=\frac14(\frac1N\log_+\L+\frac1M\log_-\L)=\sum_{i=-\infty}^{+\infty}w_i\Lambda^i,
\end{align}
which will generate a series of flow equations which contain the spatial flow in later defined Lax equations.
Here the logarithmic Lax operator $\log \L$  has the even Grassmann parity.

Under the reduction $u_{0}(x)=1$, one can derive the 1D $(2,2)$- Toda lattice hierarchy  in \cite{superbigraded}. In this case, the
representation \eqref{1dLaxEq}  will be given by the Lax operator
\bbaa\label{LaxT}
L_{2,2}=\Lambda^2+\gamma(x)\Lambda+\beta(x)+\rho(x)\Lambda^{-1}+\alpha(x)\Lambda^{-2}.\nn
\eeaa
As a consequence of eq. \eqref{1dLaxEq}, we
have the $D_0$ flow equations as

\bbaa\label{1dToda}
&&D_0 \alpha(x)=\alpha(x)(\beta(x)-\beta(x-2\epsilon)),\ \ \nn\\
&&D_0 \rho(x)=\rho(x)(\beta(x)-\beta(x-\epsilon))+\alpha(x+\epsilon) \gamma(x)-\alpha(x)\gamma(x-2\epsilon),
\nn \\
&&D_0\gamma(x)=\rho(x+2\epsilon)-\rho(x),\ \ \nn\\
 && D_0 \beta(x)=\alpha(x+2\epsilon)
-\alpha(x)+\gamma(x) \rho(x+\epsilon)+\gamma(x-\epsilon)\rho(x). \nn\eeaa

\subsection{ Lax equations of extended fermionic $(2N,2M)$-Toda hierarchy}

In this section we will  introduce the Lax equations of extended fermionic $(2N,2M)$-Toda hierarchy.
Let us first introduce some convenient notations.
\begin{definition}The operators $B_{j},G_{j}$ are defined as follows
\begin{align}\label{satoS}
\begin{aligned}
B_{j}&:=\frac{\L^{j+1}}{(j+1)!},\ \
G_{j}:=\frac{2\L^j}{j!}(\log \L-c_j),\ \  c_j=\sum_{i=1}^j\frac 1i,\ j\geq 0.
\end{aligned}
\end{align}
\end{definition}

Now we give the definition of the extended fermionic $(2N,2M)$-Toda hierarchy.
\begin{definition}The extended fermionic $(2N,2M)$-Toda hierarchy is a hierarchy in which the dressing operators $W,\bar W$ satisfy the following Sato equations
\begin{align}
\label{satoStz} \epsilon D_{t_{j}}W&=-(B_{j})_-W,& \epsilon D_{t_{j}}\bar W&=(B_{j})_+\bar W,  \\
\label{satoSsz}\epsilon D_{ y_{j}}W&=-(G_{j})_- W,& \epsilon D_{y_{j}}\bar W&=(G_{j})_+\bar W.\end{align}
\end{definition}

 From the previous proposition we derive the following  Lax equations for the Lax operators.
\begin{proposition}\label{Lax}
 The  Lax equations of the extended fermionic $(2N,2M)$-Toda hierarchy are as follows
   \begin{align}
\label{laxtjk}
  \epsilon D_{t_{j}} \L&= [(B_{j})_+,\L\},&
  \epsilon D_{y_{j}} \L&= [(G_{j})_+,\L\},\
  \epsilon D_{ t_{j}} \log \L= [(B_{j})_+ ,\log \L\},&
  \end{align}
   \begin{align}\epsilon D_{ y_{j}}\log \L=[ -(G_{j})_-,\log_+ \L \}+
[(G_{j})_+ ,\log_- \L \}.
\end{align}
\end{proposition}

\section{Additional symmetry and Block algebra}

In this section, we will  construct of the flows of additional
symmetries which form the well-known Block algebra.
With the dressing operators given $M, \bar M$, we introduce Orlov-Schulman operators as following
\begin{eqnarray}\label{Moperator}
&&M=W\Gamma W^{-1}, \ \ \bar M=\bar W\bar \Gamma \bar W^{-1},\ \\
 &&\Gamma=
\frac{x}{2N\epsilon}\Lambda^{-2N}+\sum_{n\geq 0}
\frac{\Lambda^{2Nn}}{n}t_{n}+\sum_{n\geq 0}
\frac{2}{ (n-1)!}  \Lambda^{2N(n-1)} (\epsilon \d -  c_{n-1} )y_{n},\\
&&\bar \Gamma=
\frac{-x}{2M\epsilon}\Lambda^{2M}+\sum_{n\geq 0}
\frac{2}{ (n-1)!}  \Lambda^{2M(1-n)} (-\epsilon \d -  c_{n} )y_{n}.
\end{eqnarray}

Then one can prove the Lax operator $\L$ and Orlov-Schulman operators $M,\bar M$ all have the even Grassmann parity and  satisfy the following theorem.
\begin{proposition}\label{flowsofM}
The Lax operator $\L$ and Orlov-Schulman operators $M,\bar M$ of the extended fermionic $(2N,2M)$-Toda hierarchy
satisfy the following
\begin{eqnarray}
&[\L,M\}=1,[\L,\bar M\}=1,[\log_+\L,M\}=W\Lambda^{-1} W^{-1},[\log_-\L,\bar M\}=\bar W \Lambda\bar W^{-1},\\ \label{Mequation}
&\epsilon D_{t_{k}} M^m\L^k= [(B_{k})_+,M^m\L^k\}, \epsilon D_{t_{k}} \bar M^m\L^k= [(B_{k})_+,\bar M^m\L^k\}, \\
 \notag
&\epsilon D_{y_{j}}
M^m\L^k=[\frac{\L^j}{j!}(\log_+ \L-c_j) -(G_{j})_-,
M^m\L^k\},\;  \epsilon D_{y_{j}}
\bar M^m\L^k=[-\frac{\L^j}{j!}(\log_- \L-c_j)+(G_{j})_+, \bar M^m\L^k\}.\\
\end{eqnarray}
\end{proposition}

\begin{proof}
One can prove the proposition by dressing the following several commutative Lie brackets
\begin{eqnarray*}&&[D_{ t_{n}}-\frac{\Lambda^{2(n+1)N}}{(n+1)!},\Gamma\}\\
&=&[D_{ t_{n}}-\frac{\Lambda^{2(n+1)N}}{(n+1)!},\frac{x}{2N\epsilon}\Lambda^{-2N}+\sum_{n\geq 0}
\frac{\Lambda^{2Nn}}{n}t_{n}+\sum_{n\geq 0}
\frac{2}{ (n-1)!}  \Lambda^{2N(n-1)} (\epsilon \d -  c_{n-1} )y_{n}\}\\&=&0,
\end{eqnarray*}

\begin{eqnarray*}&&[D_{ y_{n}}-\frac{2}{ n!}  \Lambda^{2Nn} (\epsilon \d -  c_n ),\Gamma\}\\
&=&[D_{ y_{n}}-\frac{2}{ n!}  \Lambda^{2Nn} (\epsilon \d -  c_n ),\frac{x}{2N\epsilon}\Lambda^{-2N}+\sum_{n\geq 0}
\frac{\Lambda^{2Nn}}{n}t_{n}+\sum_{n\geq 0}
\frac{2}{ (n-1)!}  \Lambda^{2N(n-1)} (\epsilon \d -  c_{n-1} )y_{n}\}\\&=&0,
\end{eqnarray*}

\begin{eqnarray*}&&[D_{ y_{n}}+\frac{2}{ n!}  \Lambda^{-2Mn} (-\epsilon \d -  c_n ),\bar \Gamma\}\\
&=&[D_{ y_{n}}+\frac{2}{ n!}  \Lambda^{-2Mn} (-\epsilon \d -  c_n ),\frac{-x}{2M\epsilon}\Lambda^{2M}+\sum_{n\geq 0}
\frac{2}{ (n-1)!}  \Lambda^{2M(1-n)} (-\epsilon \d -  c_{n} )y_{n}\}\\&=&0.
\end{eqnarray*}
The other identities can be proved by similar dressing methods which will be skipped here.

\end{proof}
We are now to define the additional flows, and then to
prove that they are symmetries, which are called additional
symmetries of the extended fermionic $(2N,2M)$-Toda hierarchy. We introduce additional
independent variables $t^*_{m,l}$ and define the actions of the
additional flows on the wave operators as
\begin{eqnarray}\label{definitionadditionalflowsonphi2}
\epsilon D_{t^*_{m,l}}W=-\left((M-\bar M)^m\L^l\right)_{-}W, \ \ \ \epsilon D_{t^*_{m,l}}
\bar W=\left((M-\bar M)^m\L^l\right)_{+}\bar W,
\end{eqnarray}
where $m\geq 0, l\geq 0$. The following theorem shows that the definition \eqref{definitionadditionalflowsonphi2} is compatible with reduction condition \eqref{two dressing2} of the  extended fermionic $(2N,2M)$-Toda hierarchy.
\begin{proposition}\label{preserve constraint}
The additional flows \eqref{definitionadditionalflowsonphi2} preserve reduction condition \eqref{LaxOp} or \eqref{two dressing2} of the extended fermionic $(2N,2M)$-Toda hierarchy.
\end{proposition}
\begin{proof} By performing the derivative on $\L$ dressed by $W$ and
using the additional flow about $W$ in \eqref{definitionadditionalflowsonphi2}, we get
\begin{eqnarray*}
(\epsilon D_{t^*_{m,l}}\L)&=& (\epsilon D_{t^*_{m,l}}W)\ \La^{2N} W^{-1}
+ W\ \La^{2N}\ (\epsilon D_{t_{m,l}}W^{-1})\\
&=&-((M-\bar M)^m\L^l)_{-}W\ \La^{2N}\ W^{-1}- W\ \La^{2N}
W^{-1}\ (\epsilon D_{t^*_{m,l}}W)
\ W^{-1}\\
&=&-((M-\bar M)^m\L^l)_{-} \L+ \L ((M-\bar M)^m\L^l)_{-}\\
&=&-[((M-\bar M)^m\L^l)_{-},\L\}.
\end{eqnarray*}
Similarly, we perform the derivative on $\L$ dressed by $\bar W$ and
use the additional flow about $\bar W$ in \eqref{definitionadditionalflowsonphi2} to get the following
\begin{eqnarray*}
(\epsilon D_{t^*_{m,l}}\L)&=& (\epsilon D_{t^*_{m,l}}\bar W)\ \La^{-2M} \bar W^{-1}
+ \bar W\ \La^{-2M}\ (\epsilon D_{t_{m,l}}\bar W^{-1})\\
&=&((M-\bar M)^m\L^l)_{+} \bar W\ \La^{-2M}\ \bar W^{-1}- \bar W\ \La^{-2M}
\bar W^{-1}\ (\epsilon D_{t^*_{m,l}}\bar W)
\ \bar W^{-1}\\
&=&((M-\bar M)^m\L^l)_{+} \L- \L ((M-\bar M)^m\L^l)_{+}\\
&=&[((M-\bar M)^m\L^l)_{+},\L\}.
\end{eqnarray*}
Because
\begin{eqnarray}\label{EZTHadditionalflow111.}
[M-\bar M,\L\}=0,
\end{eqnarray}
therefore
\begin{eqnarray}\label{EZTHadditionalflow1111}
\epsilon D_{t^*_{m,l}}\L=[-\left((M-\bar M)^m\L^l\right)_{-},
\L\}=[\left((M-\bar M)^m\L^l\right)_{+}, \L\},
\end{eqnarray}
which gives the compatibility of additional flow of extended fermionic $(2N,2M)$-Toda hierarchy with reduction condition \eqref{two dressing2}.
\end{proof}

Similarly,  by performing the derivative on  $M$ given in (\ref{Moperator}), there exists
a similar derivative as $\epsilon D_{t^*_{m,l}}\L$, i.e.,
\begin{eqnarray*}
(\epsilon D_{t^*_{m,l}}M)&\!\!\!=\!\!\!&(\epsilon D_{t^*_{m,l}}W)\ \Gamma  W^{-1}
+ W\ \Gamma \ (\epsilon D_{t^*_{m,l}}W^{-1})\\
&\!\!\!=\!\!\!&-((M-\bar M)^m\L^l)_{-} W\ \Gamma \ W^{-1}- W\ \Gamma
W^{-1}\ (\epsilon D_{t^*_{m,l}}W)
\ W^{-1}\\
&\!\!\!=\!\!\!&-((M-\bar M)^m\L^l)_{-} M+ M
((M-\bar M)^m\L^l)_{-}\\
&=&-[((M-\bar M)^m\L^l)_{-}, M\},
\end{eqnarray*}
where the fact that $\Gamma $ does not depend on the additional
variables $t^*_{m,l}$ has been used. Then we can take derivatives on the dressing structure of   $\bar M$ to get the additional derivatives  act on  $M$, $\bar M$ as
\begin{eqnarray}
\label{EZTHadditionalflow11'}
\epsilon D_{t^*_{m,l}}M&=&[-\left((M-\bar M)^m\L^l\right)_{-}, M\},
\\%\end{eqnarray}\begin{eqnarray}
\label{EZTHadditionalflow12}
\epsilon D_{t^*_{m,l}}
\bar M&=&[\left((M-\bar M)^m\L^l\right)_{+}, \bar M\}.
\end{eqnarray}

By  above results,  the following corollary can be easily got.
%%%%%%%%%%%%%%%%%%%%%%%%%%%%%%%%%%%%%%%%%%%%%%%%%%%%%
\begin{corollary}\label{additionflowsonLnMmAnkZ}
For $ n,k,m,l\geq 0$,
the following identities hold true
\begin{eqnarray}\label{EZTHadditionalflow4}
 \epsilon D_{t^*_{m,l}} M^n\L^k=-[((M-\bar M)^m\L^l)_{-}, M^n\L^k\}
,\ \ \
 \epsilon D_{t^*_{m,l}} \bar M^n\L^k=[((M-\bar M)^m\L^l)_{+},
 \bar M^n\L^k\}.
\end{eqnarray}
\end{corollary}
%%%%%%%%%%%%%%%%%%%%%%%%%%%%%%%%%%%%%%%%%%%%%%%%%%%%%

With  Corollary \ref{additionflowsonLnMmAnkZ}, the following theorem can be proved.

\begin{theorem}\label{symmetry}
The additional flows $D_{t^*_{m,l}}$ commute
with the extended fermionic $(2N,2M)$-Toda hierarchy flows $D_{t_{n}},D_{y_{n}}$, i.e.,
\begin{eqnarray}
[D_{t^*_{m,l}}, D_{t_{n}}\}\Phi=0,\ \  \ [D_{t^*_{m,l}}, D_{y_{n}}\}\Phi=0,
\end{eqnarray}
where $\Phi$ can be $W$, $\bar W$ or $\L$, $ n\geq 0$.

\end{theorem}
\begin{proof} According to the definition,
\begin{eqnarray*}
[D_{t^*_{m,l}},D_{t_{n}}\}W=D_{t^*_{m,l}}
(D_{t_{n}}W)-
D_{t_{n}} (D_{t^*_{m,l}}W),
\end{eqnarray*}
and using the actions of the additional flows and the fermionic $(2N,2M)$-Toda
flows on $W$,  we have
\begin{eqnarray*}
\epsilon^2[D_{t^*_{m,l}},D_{t_{n}}\}W
&=& -\epsilon D_{t^*_{m,l}}\left((B_{n})_{-}W\right)+
\epsilon D_{t_{n}} \left(((M-\bar M)^m\L^l)_{-}W \right)\\
&=& -\epsilon(D_{t^*_{m,l}}B_{n} )_{-}W-\epsilon
(B_{n})_{-}(D_{t^*_{m,l}}W)\\&&+
\epsilon[D_{t_{n}} ((M-\bar M)^m\L^l)\}_{-}W +\epsilon
((M-\bar M)^m\L^l)_{-}(D_{t_{n}}W).
\end{eqnarray*}
Using \eqref{definitionadditionalflowsonphi2} and Proposition \ref{flowsofM}, it
equals
\begin{eqnarray*}
\epsilon^2[D_{t^*_{m,l}},D_{t_{n}}\}W
&=&[\left((M-\bar M)^m\L^l\right)_{-}, B_{n}\}_{-}W+
(B_{n})_{-}\left((M-\bar M)^m\L^l\right)_{-}W\\
&&+[(B_{n})_{+},(M-\bar M)^m\L^l\}_{-}W-((M-\bar M)^m\L^l)_{-}(B_{n})_{-}W\\
&=&[((M-\bar M)^m\L^l)_{-}, B_{n}\}_{-}W- [(M-\bar M)^m\L^l,
(B_{n})_{+}\}_{-}W\\&&+
[(B_{n})_{-},((M-\bar M)^m\L^l)_{-}\}W\\
&=&0.
\end{eqnarray*}
Similarly, using \eqref{definitionadditionalflowsonphi2} and Proposition \ref{flowsofM}, we can prove the additional flows   commute with
flows $\d_{t_{n}}$ in the sense
of acting on $W$.  Here
we also give the proof for commutativity of  additional symmetries with the extended
flow $\d_{y_{n}}$. To be a little
different from the proof above, we let the Lie bracket act on $\bar W$,
\begin{eqnarray*}
\epsilon^2[D_{t^*_{m,l}},D_{y_{n}}\}\bar W &=&
\epsilon D_{t^*_{m,l}}(G_{n})_+\bar W -
\epsilon D_{y_{n}} \left(((M-\bar M)^m\L^l)_{+}\bar W  \right)\\
&=& \epsilon D_{t^*_{m,l}}(G_{n})_+ \bar W +
\epsilon(G_{n})_+ (D_{t^*_{m,l}}
)\bar W\\&&-\epsilon(D_{y_{n}} ((M-\bar M)^m\L^l))_{+}\bar W
-\epsilon((M-\bar M)^m\L^l)_{+}(D_{y_{n}}\bar W ),
\end{eqnarray*}
which further leads to
\begin{eqnarray*}
\epsilon^2[D_{t^*_{m,l}},D_{y_{n}}\}\bar W
&=&[((M-\bar M)^m\L^l)_{+}, \frac{\L^n}{n!}(\log_- \L-c_n)\}_{+}\bar W\\
&&-[((M-\bar M)^m\L^l)_{-}, \frac{\L^n}{n!}(\log_+ \L-c_n)\}_{+}\bar W\\
&&+(G_{n})_+ ((M-\bar M)^m\L^l)_{+}\bar W -((M-\bar M)^m\L^l)_{+}(G_{n})_+\bar W\\
&&-[\left(\frac{\L^n}{n!}(\log_+ \L-c_n) \right)_+-\left(\frac{\L^n}{n!}(\log_- \L-c_n) \right)_-,(M-\bar M)^m\L^l\}_{+}\bar W  \\
&=&[((M-\bar M)^m\L^l)_{+}, (\frac{\L^n}{n!}(\log_- \L-c_n))_+\}_{+}\bar W\\
&&-[((M-\bar M)^m\L^l)_{-}, \frac{\L^n}{n!}(\log_+ \L-c_n)\}_{+}\bar W\\
&&+(G_{n})_+ ((M-\bar M)^m\L^l)_{+}\bar W -((M-\bar M)^m\L^l)_{+}(G_{n})_+\bar W\\
&&+[(M-\bar M)^m\L^l,\left(\frac{\L^n}{n!}(\log_+ \L-c_n) \right)_+\}_{+}\bar W  \\
&=&0.
\end{eqnarray*} The other cases in the theorem can be proved in similar ways.
\end{proof}
The commutative property in Theorem \ref{symmetry} means that
additional flows are symmetries of the extended fermionic $(2N,2M)$-Toda hierarchy.
Since they are symmetries, we will consider the algebraic
structure among these additional symmetries which is included in the following important
theorem.
\begin{theorem}\label{WinfiniteCalgebra}
The additional flows  $\epsilon D_{t^*_{m,l}}$ of the extended fermionic $(2N,2M)$-Toda hierarchy form a Block type Lie algebra with the
following relation
 \begin{eqnarray}\label{algebra relation}
[\epsilon D_{t^*_{m,l}},\epsilon D_{t^*_{n,k}}\}= (km-n l)\epsilon D_{t^*_{m+n-1,k+l-1}},
\end{eqnarray}
which holds in the sense of acting on  $W$, $\bar W$ or $\L$ and  $m,n,l,k\geq 0.$
\end{theorem}
\begin{proof}
 By using the additional flows, we get
\begin{eqnarray*}
\epsilon^2[D_{t^*_{m,l}},D_{t^*_{n,k}}\}W&=&
\epsilon^2D_{t^*_{m,l}}(D_{t^*_{n,k}}W)-
\epsilon^2D_{t^*_{n,k}}(D_{t^*_{m,l}}W)\\
&=&-\epsilon D_{t^*_{m,l}}\left(((M-\bar M)^n\L^k)_{-}W\right)
+\epsilon D_{t^*_{n,k}}\left(((M-\bar M)^m\L^l)_{-}W\right)\\
&=&-\epsilon(D_{t^*_{m,l}}
(M-\bar M)^n\L^k)_{-}W-\epsilon((M-\bar M)^n\L^k)_{-}(D_{t^*_{m,l}} W)\\
&&+ \epsilon(D_{t^*_{n,k}} (M-\bar M)^m\L^l)_{-}W+
\epsilon((M-\bar M)^m\L^l)_{-}(D_{t^*_{n,k}} W).
\end{eqnarray*}
We further get
 \begin{eqnarray*}&&
\epsilon^2[D_{t^*_{m,l}},D_{t^*_{n,k}}\}W\\
&=&-\epsilon\Big[\sum_{p=0}^{n-1}
(M-\bar M)^p(D_{t^*_{m,l}}(M-\bar M))(M-\bar M)^{n-p-1}\L^k
+(M-\bar M)^n(D_{t^*_{m,l}}\L^k)\Big]_{-}W\\&&-((M-\bar M)^n\L^k)_{-}(D_{t^*_{m,l}} W)\\
&&+\epsilon\Big[\sum_{p=0}^{m-1}
(M-\bar M)^p(D_{t^*_{n,k}}(M-\bar M))(M-\bar M)^{m-p-1}\L^l
+(M-\bar M)^m(D_{t^*_{n,k}}\L^l)\Big]_{-}W\\&&+
\epsilon((M-\bar M)^m\L^l)_{-}(D_{t^*_{n,k}} W)\\
&=&[(nl-km)(M-\bar M)^{m+n-1}\L^{k+l-1}]_-W\\
&=&(km-nl)\epsilon D_{t^*_{m+n-1,k+l-1}}W.
\end{eqnarray*}
Similarly  the same results on acting on $\bar W$ and $\L$ are as follows
 \begin{eqnarray*}
\epsilon^2[D_{t^*_{m,l}},D_{t^*_{n,k}}\}\bar W
&=&((km-nl)(M-\bar M)^{m+n-1}\L^{k+l-1})_+\bar W\\
&=&(km-nl)\epsilon D_{t^*_{m+n-1,k+l-1}}\bar W,
\\[6pt]%\end{eqnarray*}\begin{eqnarray*}
\epsilon^2[D_{t^*_{m,l}},D_{t^*_{n,k}}\}\L&=&
\epsilon^2D_{t^*_{m,l}}(D_{t^*_{n,k}}\L)-
\epsilon^2D_{t^*_{n,k}}(D_{t^*_{m,l}}\L)\\
&=&[((nl-km)(M-\bar M)^{m+n-1}\L^{k+l-1})_-, \L\}\\
&=&(km-nl)\epsilon D_{t^*_{m+n-1,k+l-1}}\L.
\end{eqnarray*}
\end{proof}

\vskip 0.5truecm \noindent{\bf Acknowledgments.}

This work is supported by the National Natural Science Foundation of China under Grant No.  11571192 and K. C. Wong Magna Fund in
Ningbo University.

\end{document}